\documentclass[12pt,reqno,fleqn]{amsart}
\usepackage{latexsym}
\usepackage{amssymb}
\usepackage{amsxtra}
\usepackage{amsmath}
\usepackage{amsfonts}
\usepackage[dvips]{graphicx}
\usepackage{amsmath}
\usepackage{amssymb,amsmath,amsthm,amscd,mathrsfs,amsbsy,amsfonts}
\usepackage{pstricks,pst-node}
\usepackage{amsbsy,young,colortbl,cite}

\usepackage{graphicx}
\usepackage[metapost]{mfpic}
\usepackage{mflogo}

\usepackage{epstopdf}
\usepackage{graphics}
\makeatletter      
\@addtoreset{equation}{section}
\makeatother       

\def\@makecaption#1#2{\vskip\abovecaptionskip
  \sbox\@tempboxa{\small #1: #2}%
  \ifdim \wd\@tempboxa >\hsize \small #1: #2\par
  \else \global \@minipagefalse \hb@xt@\hsize{\hfil\box\@tempboxa\hfil}\fi
  \vskip\belowcaptionskip}
\newtheorem{theorem}{Theorem}[section]

\newtheorem{proposition}[theorem]{Proposition}
\newtheorem{corollary}[theorem]{Corollary}
\newtheorem{definition}[theorem]{Definition}

\theoremstyle{remark}

\renewcommand{\L}{\mathcal{L}}

\def\res{\mathop{\rm Res}\nolimits}
\def\({\left(}
\def\){\right)}
\def\[{\begin{eqnarray}}
\def\]{\end{eqnarray}}
\def\d{\partial}

\def\La{\Lambda}

\def\al{\alpha}

\begin{document}

\title{Additional symmetry of the modified extended Toda hierarchy }
\author{
ChuanZhong Li\dag,\  \ Jingsong He\ddag}
\allowdisplaybreaks
\dedicatory {\small Department of Mathematics,  Ningbo University, Ningbo, 315211, China\\
\dag lichuanzhong@nbu.edu.cn\\
\ddag hejingsong@nbu.edu.cn}
\thanks{\ddag Corresponding author}

\texttt{}

\date{}
\begin{abstract}

  In this paper, one new integrable modified extended Toda hierarchy(METH) is constructed with the help of two logarithmic Lax operators.  With this modification, the interpolated spatial flow is added to make all flows complete. To show more integrable properties of  the METH,  the bi-Hamiltonian structure and tau symmetry of the METH will be given.
 The  additional symmetry flows of this new hierarchy are presented.
These flows form an infinite dimensional Lie algebra of Block type.

\end{abstract}
\maketitle
\allowdisplaybreaks
\maketitle{\allowdisplaybreaks}
{\small \noindent {\bf Mathematics Subject Classifications (2000)}:  37K05, 37K10, 37K20, 17B65, 17B67.

}
 \setcounter{section}{0}
\section{Introduction}
For integrable dynamical systems, the study of their symmetries plays a central role in their development.
In particular, the additional symmetry of nonlinear dynamical systems is one important subject in mathematics and physics.
 It  can lead to string equations and Virasoro constraints which
are important and extensively involved in the matrix models of the string theory\cite{Morozov1,Aratyn,Morozov2,Morozov3}.
Additional symmetries  of KP hierarchy were given by Orlov and
Shulman \cite{os1} through two novel operators $\Gamma$ and $M$, which can be used
to form a centerless  $W$ algebra. It is well known that there are two kinds of sub-hierarchies of KP, the BKP hierarchy
of which Virasoro constraints
and the ASvM formula were constructed in \cite{vandeLeur,vandeLeur2,tu,tu2} and a CKP hierarchy of which additional symmetries and string equation
were given in \cite{heLMP,kelei}.
Comparing with the KP hierarchy, the Toda type systems as important differential-difference hierarchies were proved to have the similar additional symmetry.
The
two dimensional Toda hierarchy
was constructed  by  Ueno and Takasaki \cite{UT} with the
help of difference
operators and infinite dimensional Lie algebras. The additional symmetry of the two dimensional Toda hierarchy was studied using ASvM formula in \cite{asv2}.
The one dimensional Toda hierarchy (TH) was also studied under a
 certain reduction condition  on two Lax operators. The bigraded Toda hierarchy (BTH) of $(N,M)$-type is the generalized Toda hierarchy
whose infinite Lax matrix has $N$ upper and $M$ lower nonzero diagonals\cite{solutionBTH}. One interesting thing is that it is found that the additional symmetry of the BTH composed a novel Block type algebraic structure. From above, it seems that the additional symmetry is one universal structure for integrable systems.

Besides above integrable systems related to the Toda system, there exist several kinds of extension of Toda hierarchy which is expressed with extended logarithmic flows. For example, recently in
\cite{CDZ}, the interpolated Toda lattice
hierarchy was generalized to the so-called extended Toda hierarchy (ETH) for considering
its application on the topological field theory. Later the Hirota quadratic equations of the ETH was constructed in \cite{M} with the help of generalized Vertex operators. In \cite{C}, the TH and ETH were further
generalized to the extended bigraded Toda hierarchy (EBTH) by considering $N+M$ dependent variables
in the Lax operator.
 In \cite{TH},
the Hirota bilinear equations (HBEs) of EBTH have been given
conjecturally and proved that it governs the Gromov-Witten theory of
orbifold $c_{km}$. In \cite{ourJMP}, the authors generalize the
Sato theory to the EBTH and give the Hirota bilinear equations in
terms of vertex operators whose coefficients  take values in the
algebra of differential operators. The Hirota bilinear equation of EBTH were equivalently constructed in a very recent paper \cite{leurhirota}, because of the equivalence of $t_{1,N}$ flow and $t_{0,N}$ flow of EBTH in \cite{ourJMP}. Meanwhile it was proved to govern Gromov-Witten invariant of the total
  descendent potential of $\mathbb{P}^1$ orbifolds \cite{leurhirota} which finished the conjecture proposed in \cite{TH}.
But from the point of additional symmetry, it is a pity that these two kinds of hierarchies which both contain logarithmic extended flows have no natural additional
 symmetry using the standard constructing method in \cite{os1} because of the single logarithmic operator.
Because of the universality of additional symmetry for integrable systems and the importance of extended flows in Gromov-Witten theory, in this paper, we will modify the ETH and construct one new kind of integrable system named as the modified extended Toda hierarchy (METH).  We further prove that this new hierarchy have additional symmetries which also form a nice Block Lie algebra.
Infinite dimensional Lie algebras of Block type,
as generalizations of the well-known Virasoro algebra,  have
been studied intensively in literature  \cite{Block,DZ,Su}.  In \cite{ourBlock}, we provide  this kind of Block type algebraic structure for the bigraded Toda hierarchy (BTH) \cite{ourJMP,solutionBTH}.
Later on, this Block type Lie algebra is found again in dispersionless bigraded Toda
 hierarchy \cite{dispBTH}.

This paper is arranged as follows. In Section 2 we define the modified extended Toda hierarchy
with two separate logarithms of the Lax operator $\L$.
In the Section 3, the bi-Hamiltonian structure and tau symmetry of the METH will be given to prove the integrability of the METH.  The additional
symmetries will be constructed in Section 4 and they form the well-known Block algebra. At last, we will give further conclusion and discussion.

\section{The modified extended Toda hierarchy}
Set a Lax operator be the following Laurent polynomial \cite{CDZ}
\begin{equation}\label{LMETH}
\L:=\Lambda+u+e^{v(x)}\Lambda^{-1},
\end{equation}
where the shift operator $\Lambda$ acts on a function $a(x)$ by $\Lambda a(x) = a(x + \epsilon)$,
i.e. $\Lambda$ is equivalent to  $e^{\epsilon\partial_{x}}$ where
the spacing unit $``\epsilon"$ is called string coupling constant.
 The Lax operator $\L$ can be
written in two different ways by dressing the shift operator
\begin{equation}\label{two dressing}
\L=S\Lambda S^{-1} = \bar S^{-1} \Lambda^{-1}\bar S^{-1},
\end{equation}
where the dressing operators have the form,
 \begin{align}
 S&=1+w_1\Lambda^{-1}+w_2\Lambda^{-2}+\ldots,
\label{dressP}\\[0.5ex]
 \bar S&=\tilde{w_0}+\tilde{w_1}\Lambda+\tilde{w_2}\Lambda^2+ \ldots.
\label{dressQ}
\end{align}
 The
pair is unique up to multiplying $S$ and $\bar S$ from the right
 by operators in the form  $1+
a_1\Lambda^{-1}+a_2\Lambda^{-2}+...$ and $\tilde{a}_0 +
\tilde{a}_1\Lambda +\tilde{a}_2\Lambda^2+\ldots$ respectively with
coefficients independent of $x$. Given any difference operator $A=
\sum_k A_k \Lambda^k$, the positive and negative projections are
defined by $A_+ = \sum_{k\geq0} A_k \Lambda^k$ and $A_- = \sum_{k<0}
A_k \Lambda^k$. To construct the METH, we define the following logarithm operators
\begin{align}
\log_+\L&=(S\cdot\epsilon \partial\cdot S^{-1})= \epsilon \partial+\sum_{k < 0} W_k(x) \Lambda^k=-\log (1+(\Lambda^{-1}-1))+\sum_{k < 0} W_k(x) \Lambda^k,\\
\log_-\L&=-(\bar S\cdot\epsilon \partial\cdot \bar S^{-1})=-\epsilon \partial+\sum_{k \geq 0} W_k(x) \Lambda^k=-\log (1+(\Lambda-1))+\sum_{k \geq 0} W_k(x) \Lambda^k ,
\end{align}
where $\d$ is the derivative about spatial variable $x$.
Therefore two logarithmic operators can be rewritten in the following forms
\[\log_+\L&=\sum_{k \leq 0} \tilde W_k(x) \Lambda^k,\ \
\log_-\L=\sum_{k \geq 0} \tilde W_k(x) \Lambda^k.\]

\begin{definition} \label{deflax}
The Lax formulation of modified extended  Toda hierarchy can be given
by
\begin{equation}
  \label{edef}
\frac{\partial \L}{\partial t_{\alpha, n}} = [ A_{\alpha,n} ,\L ],
\end{equation}
for $\alpha =0,1,2$ and $n \geq 0$. The operators
$A_{\alpha ,n}$ are defined by
\begin{subequations}
\label{Adef}
\begin{align}
  &A_{0,n} = \frac{1}{(n+1)!}\L^{n+1}_+,\\
  &A_{1,n} = \frac{2}{ n!} [ \L^n (\log_+ \L -  c_n ) ]_+, \\
  &A_{2,n} =-\frac{2}{(n+1)!} [\L^{n+1} (\log_- \L - c_{n+1} ) ]_- .
\end{align}
\end{subequations}
The constants $c_n$ are defined by
\begin{equation}
  \label{b24}
  c_n = \sum_{k=1}^n \frac1k ,  c_0=0 .
\end{equation}
\end{definition}

{\bf Remark 1:}
The difference between the METH and the ETH \cite{CDZ} is that we separate two different logarithmic operators to construct two independent extended flows. For the ETH, the extended flow (see $t^{1,q}$ flow in eq.(2.27) of \cite{CDZ})is constructed with a single logarithm which is a combination of $\log_+ \L$ and $\log_- \L$. But this kind of construction becomes an obstacle to construct additional symmetry. That is why we construct two separated logarithmic flows for the METH. That combination exactly canceled the $\d$ operator technically. The $t_{1, 0}$ flow of the METH is exactly the spatial flow which make all the flows complete. This is the same as the ETH.

For the convenience, we give following notations
\begin{align}
  &B_{0,n} = \frac{1}{(n+1)!}\L^{n+1},\\
  &B_{1,n} = \frac{2}{ n!}  \L^n (\log_+ \L -  c_n ), \\
  &B_{2,n} =-\frac{2}{(n+1)!} \L^{n+1} (\log_- \L - c_{n+1} ) .
\end{align}
One can show that the METH in the Lax representation can be written in the equations
of the dressing operators (i.e. the Sato equations):
\begin{proposition}
\label{t1} The operator $\L$ in \eqref{LMETH} is a solution to the METH \eqref{edef} if and only if there is a
pair of dressing operators $S$ and $\bar S$ which satisfy the Sato  equations,
\begin{equation}
\label{bn1}
\d_{\gamma,n}S  =- (B_{\gamma,n})_-  S, \qquad  \d_{\gamma,n}\bar S = (B_{\gamma ,n})_+\bar S
\end{equation}
for $0\leq\gamma\leq 2$ and  $n \geq 0.$
\end{proposition}

In the next section, to to show more integrable properties of  the METH,  the bi-Hamiltonian structure and tau symmetry of the METH will be given.

\section{Bi-Hamiltonian structure and tau symmetry}

The bi-Hamiltonian structure for the
modified extended Toda hierarchy can be given by the following two compatible Poisson brackets which has the same form as extended Toda hierarchy in \cite{CDZ}

\begin{eqnarray}
&&\{v(x),v(y)\}_1=\{u(x),u(y)\}_1=0,\notag\\
&&\{u(x),v(y)\}_1=\frac{1}{\epsilon} \left[e^{\epsilon\,\d_x}-1
\right]\delta(x-y),\label{toda-pb1}\\
&& \{u(x),u(y)\}_2={1\over\epsilon}\left[e^{\epsilon\,\d_x}
e^{v(x)}-
e^{v(x)} e^{-\epsilon\,\d_x}\right] \delta(x-y),\notag\\
&& \{ u(x), v(y)\}_2 = {1\over \epsilon}
u(x)\left[e^{\epsilon\,\d_x}-1 \right]
\delta(x-y)\label{toda-pb2}\\
&& \{ v(x), v(y)\}_2 = {1\over \epsilon} \left[
e^{\epsilon\,\d_x}-e^{-\epsilon\,\d_x}\right]\delta(x-y).\notag
\end{eqnarray}

For any difference operator $A=
\sum_k A_k \Lambda^k$, define residue $\res A=A_0$.
In the following theorem, we will prove the above Poisson structure can be as the the Hamiltonian structure of the METH.
\begin{theorem}
The flows of the modified extended Toda hierarchy  are Hamiltonian systems
of the form
\begin{equation}
\frac{\d u}{\d t_{\beta,n}}=\{u,H_{\beta,n}\}_1, \  \frac{\d v}{\d t_{\beta,n}}=\{v,H_{\beta,n}\}_1,\quad \beta=0,1,2;\ n\ge 0.
\label{td-ham}
\end{equation}
They satisfy the following bi-Hamiltonian recursion relation
\[\notag
\{\cdot,H_{1,n-1}\}_2&=&n
\{\cdot,H_{1,n}\}_1+2\{\cdot,H_{0,n-1}\}_1,\ \{\cdot,H_{0,n-1}\}_2=(n+1)
\{\cdot,H_{0,n}\}_1\\ \label{recursion}
\{\cdot,H_{2,n-1}\}_2&=&n
\{\cdot,H_{2,n}\}_1+2\{\cdot,H_{2,n-1}\}_1.
\]
Here the Hamiltonians have the form
\begin{equation}
H_{\beta,n}=\int h_{\beta,n}(u,v; u_x,v_x; \dots; \epsilon) dx,\quad \beta=0,1,2; \ n\ge 0,
\end{equation}
with the Hamiltonian densities $h_{\beta,n}=h_{\beta,n}(u,v; u_x,v_x; \dots; \epsilon)$ given by
\[
 h_{0,n}&=&\frac1{(n+1)!}\res \, \L^{n+1},\ h_{1,n}=\frac2{n!}\,\res\left[ \L^{n}
(\log_+ \L-c_{n})\right],\\
 h_{2,n}&=&\frac2{n!}\,\res\left[ \L^{n}
(\log_- \L-c_{n})\right].
\]

\end{theorem}

\begin{proof}
For $\beta=0$, i.e. the original Toda hierarchy, the proof is same as the proof in \cite{CDZ}.

Here we will prove that the flows $\frac{\d}{\d t_{1,n}}$ are also
Hamiltonian systems with respect to the first Poisson bracket.
In \cite{CDZ}, the following identity has been proved
\begin{equation}\label{dlgl-2} \res\left[\L^n d (S\epsilon \d_x
S^{-1})\right] \sim \res \L^{n-1} d \L,
\end{equation}
which show the validity of the following equivalence relation:
\begin{equation}\label{dlgl}
\res\left(\L^n\, d \log_+ \L\right) \sim \res\left(\L^{n-1} d \L\right).
\end{equation}
Here the equivalent relation $\sim$ is up to a $x$-derivative of
another 1-form.

In a similar way as eq.\eqref{dlgl-2}, we obtain the following equivalence relation
\begin{equation}\label{dlgl-3}
\res\left[\L^n d (\bar S\epsilon\d_x \bar S^{-1})\right]\sim -\rm \res \L^{n-1} d \L£¬
\end{equation}
i.e.
\begin{equation}\label{dlgl}
\res\left(\L^n\, d \log_- \L\right) \sim \rm \res\left(\L^{n-1} d \L\right).
\end{equation}
Suppose
\[
A_{\alpha,n}=\sum_{k} a_{\alpha,n+1;k}\, \Lambda^k,
\]
Then from
\begin{equation}
  \label{edef3}
\frac{\partial \L}{\partial t_{\alpha, n}} = [ (B_{\alpha,n})_+ ,\L ]= [ -(B_{\alpha,n})_- ,\L ],
\end{equation}
we can derive equation
\[\epsilon\frac{\partial u}{\partial t_{\beta, n}}&=&a_{\beta,n;1}(x+\epsilon)-a_{\beta,n;1}(x)\\
\epsilon\frac{\partial v}{\partial t_{\beta, n}}&=&a_{\beta,n;0}(x-\epsilon) e^{v(x)}-a_{\beta,n;0}(x) e^{v(x+\epsilon)}.
\]

The equivalence relation (\ref{dlgl}) now readily follows from the above two equations.
By using (\ref{dlgl}) we obtain
\begin{eqnarray}
&&d h_{1,n}=\frac2{n!}\,d\,\res\left[\L^{n}
\left(\log_+ \L-c_{n}\right) \right]
\notag\\
&& \sim \frac2{(n-1)!}\,\res\left[\L^{n-1}
\left(\log_+ L-c_{n}\right) d L\right]+ \frac2{n!}\,\res\left[\L^{n-1} d \L\right]\notag\\
&&=\frac2{(n-1)!}\,\res\left[\L^{n-1} \left(\log_+ \L-c_{n-1}\right) d \L\right]\\
&&=\res\left[a_{1,n;0}(x)du+a_{1,n;1}(x-\epsilon) e^{v(x)}dv\right].
\end{eqnarray}
It yields the following identities
\begin{equation}\label{dH1-u12}
\frac{\delta H_{1,n}}{\delta u}=a_{1,n;0}(x),\quad \frac{\delta H_{1,n}}
{\delta v}=a_{1,n;1}(x-\epsilon) e^{v(x)}.
\end{equation}
This agree with Lax equation

\[
\frac{\d u}{\d t_{1,n}}&=&\{u,H_{1,n}\}_1={1\over \epsilon} \left[
e^{\epsilon\,\d_x}-1\right]\frac{\delta H_{1,n}}
{\delta v}={1\over \epsilon}(a_{1,n;1}(x+\epsilon)-a_{1,n;1}(x)),\\
 \  \frac{\d v}{\d t_{1,n}}&=&\{v,H_{1,n}\}_1=\frac{1}{\epsilon} \left[1-e^{\epsilon\,\d_x}
\right]\frac{\delta H_{1,n}}
{\delta u}=\frac{1}{\epsilon} \left[a_{1,n;0}(x-\epsilon) e^{v(x)}-a_{1,n;0}(x) e^{v(x+\epsilon)}\right].
\]

 From the above identities we see that
the flows $\frac{\d}{\d t_{1,n}}$ are Hamiltonian systems
of the form (\ref{td-ham}).
For the case of $\beta=1,2$ the recursion relation (\ref{recursion})
follows from the following trivial identities
\begin{eqnarray}
&&n\, \frac{2}{n!} \L^{n} \left(\log_{\pm} \L-c_{n}\right)=\L\,
\frac{2}{(n-1)!}
\L^{n-1} \left(\log_{\pm} \L-c_{n-1}\right)-2\,\frac1{n!} \L^n\notag\\
&&=\frac{2}{(n-1)!} \L^{n-1} \left(\log_{\pm} \L-c_{n-1}\right)\,
\L-2\,\frac1{n!} \L^n.\notag
\end{eqnarray}
Then we get, for $\beta=1,2,$
\begin{eqnarray}
&&n a_{\beta,n+1;1}(x)=a_{\beta,n;0}(x+\epsilon)+ua_{\beta,n;1}(x)+e^va_{\beta,n;2}(x-\epsilon)-2a_{0,n+1;1}(x)\notag\\
&&=a_{\beta,n;0}(x)+u(x+\epsilon)a_{\beta,n;1}(x)+e^{v(x+2\epsilon)}a_{\beta,n;2}(x)-2a_{0,n+1;1}(x).\notag
\end{eqnarray}
This further leads to

\begin{eqnarray}
&&\{u,H_{\beta,n-1}\}_2=\left[\Lambda e^{v(x)}-e^{v(x)} \Lambda^{-1}\right] a_{\beta,n;0}(x)+
u(x) \left[\Lambda-1\right] a_{\beta,n;1}(x-\epsilon) e^{v(x)}\notag\\ \notag
&&
=n\left[a_{\beta,n+1;1}(x) e^{v(x+\epsilon)}-a_{\beta,n+1;1}(x-\epsilon) e^{v(x)}\right]+2\left[a_{0,n+1;0}(x) e^{v(x+\epsilon)}-a_{0,n+1;0}(x-\epsilon) e^{v(x)}\right].\label{pre-recur}
\end{eqnarray}
This is exactly the recursion relation eq.\eqref{recursion} for $u$. The similar recursion flow on $v$ can be similarly derived.
Theorem is proved till now.

\end{proof}

{\bf Remark 2:}
The difference between the bi-Hamiltonian structure of the METH and the one of the ETH \cite{CDZ} is that we separate the Hamiltonian $h_{1,q}$ in eq.(3.8) of \cite{CDZ} into two different logarithmic Hamiltonians $h_{1,n},h_{2,n}$ to construct two independent extended flows of the METH.

Similarly as \cite{CDZ}, the tau symmetry of the METH can be proved in the  following theorem.
\begin{theorem}\label{tausymmetry}
The modified extended Toda hierarchy has the following tau-symmetry property:
\begin{equation}
\frac{\d h_{\alpha,m}}{\d t_{\beta,n}}=\frac{\d
h_{\beta,n}}{\d t_{\alpha,m}},\quad \alpha,\beta=0,1,2,\ m,n\ge 0.
\end{equation}
\end{theorem}
\begin{proof} Let us prove the theorem for the case when $\alpha=1, \beta=0$,
other cases are proved in a similar way
\[
&&\frac{\d h_{1,m}}{\d t_{0,n}} =\frac2{m!\,(n+1)!}\,\res[-(L^{n+1})_-, L^m (\log_+L-c_m)]\notag\\
&&=\frac2{m!\,(n+1)!}\,\res[(L^m (\log_+L-c_m))_+,(L^{n+1})_-]\notag\\
&& =\frac2{m!\,(n+1)!}\,\res[(L^m (\log_+L
-c_m))_+,L^{n+1}]=\frac{\d h_{0,n}}{\d t_{1,m}}.
\]
Theorem is proved.
\end{proof}

 This property justifies the following definition of
tau function for the modified extended Toda hierarchy:

\begin{definition} The  the $tau$ function of the extened Toda hierarchy can be defined by
the following expressions in terms of the densities of the Hamiltonians:
\begin{equation}
h_{\alpha,n}=\epsilon (\Lambda-1)\frac{\d\log\tau}{\d t_{\alpha,n+1}},\quad \al=0,1,2,\ n\ge 0.
\end{equation}
\end{definition}
The existence of tau function of the METH can be easily proved by Theorem \ref{tausymmetry}.
Note that
the above definition implies in particular the following relations of the
dependent variables $v, u$ of the modified extended Toda hierarchy with the tau function:
\begin{equation}
v=\epsilon \frac{\d}{\d t_{0,0}}\log \frac{\tau(x+\epsilon)}{\tau(x)},\quad
u=\log\frac {\tau(x+\epsilon)\tau(x-\epsilon)}{\tau^2(x)}.
\end{equation}

The above section tells us that the METH is an integrable system. The next section will be devoted to the additional symmetry of the METH.

\section{Additional symmetry and Block algebra}

In this section, we will put constrained condition
eq.\eqref{two dressing} into construction of the flows of additional
symmetries which form the well-known Block algebra.

With the dressing operators given in eq.\eqref{two dressing}, we introduce Orlov-Schulman operators as following
\begin{eqnarray}\label{Moperator}
&&M=S\Gamma S^{-1}, \ \ \bar M=\bar S\bar \Gamma \bar S^{-1},\ \\
 &&\Gamma=
\frac{x}{\epsilon}\Lambda^{-1}+\sum_{n\geq 0}
(n+1)\Lambda^{n}t_{0,n}+\sum_{n\geq 0}
\frac{2}{ (n-1)!}  \Lambda^{n-1} (\epsilon \d -  c_{n-1} )t_{1,n},\\
&&\bar \Gamma=
\frac{-x}{\epsilon}\Lambda+\sum_{n\geq 0}
\frac{2}{ n!}  \Lambda^{-n} (-\epsilon \d -  c_{n} )t_{2,n}.
\end{eqnarray}

Then one can prove the Lax operator $\L$ and Orlov-Schulman operators $M,\bar M$ satisfy the following theorem.
\begin{proposition}\label{flowsofM}
The Lax operator $\L$ and Orlov-Schulman operators $M,\bar M$ of the METH
satisfy the following
\begin{eqnarray}
&[\L,M]=1,[\L,\bar M]=1,[\log_+\L,M]=\L,[\log_-\L,\bar M]=\L\\ \label{Mequation}
&\partial_{ t_{\alpha,n}}M=
[A_{\alpha,n},M],\ \ \partial_{ t_{\alpha,n}}\bar M=[A_{\alpha,n},\bar M],\\
&\dfrac{\partial
M^m\L^k}{\partial{t_{\alpha,n}}}=[A_{\alpha,n},
M^m\L^k],\;  \dfrac{\partial
\bar M^m\L^k}{\partial{t_{\alpha,n}}}=[A_{\alpha,n}, \bar M^m\L^k]
\end{eqnarray}

\end{proposition}

\begin{proof}
One can prove the proposition by dressing the following several commutative Lie brackets
\begin{eqnarray*}&&[\partial_{ t_{0,n}}-\frac{\Lambda^{n+1}}{(n+1)!},\Gamma]\\
&=&[\partial_{ t_{0,n}}-\frac{\Lambda^{n+1}}{(n+1)!},\frac{x}{\epsilon}\Lambda^{-1}+\sum_{n\geq 0}
\frac{\Lambda^{n}}{n!}t_{0,n}+\sum_{n\geq 0}
\frac{2}{ (n-1)!}  \Lambda^{n-1} (\epsilon \d -  c_{n-1} )t_{1,n}]\\&=&0,
\end{eqnarray*}

\begin{eqnarray*}&&[\partial_{ t_{0,n}}-\frac{2}{ n!}  \Lambda^n (\epsilon \d -  c_n ),\Gamma]\\
&=&[\partial_{ t_{0,n}}-\frac{2}{ n!}  \Lambda^n (\epsilon \d -  c_n ),\frac{x}{\epsilon}\Lambda^{-1}+\sum_{n\geq 0}
(n+1)\Lambda^{n}t_{0,n}+\sum_{n\geq 0}
\frac{2}{ (n-1)!}  \Lambda^{n-1} (\epsilon \d -  c_{n-1} )t_{1,n}]\\&=&0,
\end{eqnarray*}

\begin{eqnarray*}&&[\partial_{ t_{0,n}}+\frac{2}{ n!}  \Lambda^{-n} (-\epsilon \d -  c_n ),\bar \Gamma]\\
&=&[\partial_{ t_{0,n}}+\frac{2}{ n!}  \Lambda^{-n} (-\epsilon \d -  c_n ),\frac{-x}{\epsilon}\Lambda+\sum_{n\geq 0}
\frac{2}{ n!}  \Lambda^{-n} (-\epsilon \d -  c_{n} )t_{2,n}]\\&=&0.
\end{eqnarray*}

\end{proof}
{\bf Remark 3:}
The reason why we modified the ETH is in the obstacle to construct suitable Orlov-Schulman operators $M,\bar M$ to satisfy eq.\eqref{Mequation}.
For the ETH, we failed to do this because of the combination of $\log_+\L$ and $\log_-\L$  in the extended flows.

We are now to define the additional flows, and then to
prove that they are symmetries, which are called additional
symmetries of the METH. We introduce additional
independent variables $t^*_{m,l}$ and define the actions of the
additional flows on the wave operators as
\begin{eqnarray}\label{definitionadditionalflowsonphi2}
\dfrac{\partial S}{\partial
{t^*_{m,l}}}=-\left((M-\bar M)^m\L^l\right)_{-}S, \ \ \ \dfrac{\partial
\bar S}{\partial {t^*_{m,l}}}=\left((M-\bar M)^m\L^l\right)_{+}\bar S,
\end{eqnarray}
where $m\geq 0, l\geq 0$. The following theorem shows that the definition \eqref{definitionadditionalflowsonphi2} is compatible with reduction condition \eqref{two dressing} of the  METH.
\begin{proposition}\label{preserve constraint}
The additional flows \eqref{definitionadditionalflowsonphi2} preserve reduction condition \eqref{two dressing}.
\end{proposition}
\begin{proof} By performing the derivative on $\L$ dressed by $S$ and
using the additional flow about $S$ in \eqref{definitionadditionalflowsonphi2}, we get
\begin{eqnarray*}
(\partial_{t^*_{m,l}}\L)&=& (\partial_{t^*_{m,l}}S)\ \La S^{-1}
+ S\ \La\ (\partial_{t_{m,l}}S^{-1})\\
&=&-((M-\bar M)^m\L^l)_{-} S\ \La\ S^{-1}- S\ \La
S^{-1}\ (\partial_{t^*_{m,l}}S)
\ S^{-1}\\
&=&-((M-\bar M)^m\L^l)_{-} \L+ \L ((M-\bar M)^m\L^l)_{-}\\
&=&-[((M-\bar M)^m\L^l)_{-},\L].
\end{eqnarray*}
Similarly, we perform the derivative on $\L$ dressed by $\bar S$ and
use the additional flow about $\bar S$ in \eqref{definitionadditionalflowsonphi2} to get the following
\begin{eqnarray*}
(\partial_{t^*_{m,l}}\L)&=& (\partial_{t^*_{m,l}}\bar S)\ \La \bar S^{-1}
+ \bar S\ \La\ (\partial_{t_{m,l}}\bar S^{-1})\\
&=&((M-\bar M)^m\L^l)_{+} \bar S\ \La^{-1}\ \bar S^{-1}- \bar S\ \La
\bar S^{-1}\ (\partial_{t^*_{m,l}}\bar S)
\ \bar S^{-1}\\
&=&((M-\bar M)^m\L^l)_{+} \L- \L ((M-\bar M)^m\L^l)_{+}\\
&=&[((M-\bar M)^m\L^l)_{+},\L].
\end{eqnarray*}
Because
\begin{eqnarray}\label{METHadditionalflow111.}
[M-\bar M,\L]=0,
\end{eqnarray}
therefore
\begin{eqnarray}\label{METHadditionalflow1111}
\dfrac{\partial \L}{\partial
{t^*_{m,l}}}=[-\left((M-\bar M)^m\L^l\right)_{-},
\L]=[\left((M-\bar M)^m\L^l\right)_{+}, \L],
\end{eqnarray}
which gives the compatibility of additional flow of METH with reduction condition \eqref{two dressing}.
\end{proof}

Similarly, we can take derivatives on the dressing structure of  $M$ and  $\bar M$ to get the following proposition.
\begin{proposition}\label{add flow}
The additional derivatives  act on  $M$, $\bar M$ as
\begin{eqnarray}
\label{METHadditionalflow11'}
\dfrac{\partial
M}{\partial{t^*_{m,l}}}&=&[-\left((M-\bar M)^m\L^l\right)_{-}, M],
\\
\label{METHadditionalflow12}
\dfrac{\partial
\bar M}{\partial{t^*_{m,l}}}&=&[\left((M-\bar M)^m\L^l\right)_{+}, \bar M].
\end{eqnarray}
\end{proposition}
\begin{proof} By performing the derivative on  $M$ given in (\ref{Moperator}), there exists
a similar derivative as $\partial_{t^*_{m,l}}\L$, i.e.,
\begin{eqnarray*}
(\partial_{t^*_{m,l}}M)&\!\!\!=\!\!\!&(\partial_{t^*_{m,l}}S)\ \Gamma  S^{-1}
+ S\ \Gamma \ (\partial_{t^*_{m,l}}S^{-1})\\
&\!\!\!=\!\!\!&-((M-\bar M)^m\L^l)_{-} S\ \Gamma \ S^{-1}- S\ \Gamma
S^{-1}\ (\partial_{t^*_{m,l}}S)
\ S^{-1}\\
&\!\!\!=\!\!\!&-((M-\bar M)^m\L^l)_{-} M+ M
((M-\bar M)^m\L^l)_{-}\\
&=&-[((M-\bar M)^m\L^l)_{-}, M].
\end{eqnarray*}
Here the fact that $\Gamma $ does not depend on the additional
variables $t^*_{m,l}$ has been used. Other identities can also
be obtained in a similar way.
\end{proof}

By two propositions above,  the following corollary can be easily got.
\begin{corollary}\label{additionflowsonLnMmAnk}
For $ n,k,m,l\geq 0$,
the following identities hold
\begin{eqnarray}\label{METHadditionalflow4}
 \dfrac{\partial M^n\L^k}{\partial{t^*_{m,l}}}=-[((M-\bar M)^m\L^l)_{-}, M^n\L^k]
,\ \ \
 \dfrac{\partial \bar M^n\L^k}{\partial{t^*_{m,l}}}=[((M-\bar M)^m\L^l)_{+},
 \bar M^n\L^k].
\end{eqnarray}
\end{corollary}
\begin{proof}  First we present the proof of the first
equation.   Considering the dressing relation
\begin{eqnarray*}
 \L^{ n}=S\Lambda^nS^{-1},
\end{eqnarray*}
and \eqref{definitionadditionalflowsonphi2}, we can get relations
\begin{eqnarray*}
\dfrac{\partial \L^n}{\partial{t^*_{m,l}}}=[-((M-\bar M)^m\L^l)_{-}, \L^n],
\end{eqnarray*}
which further leads to the second identity in
\eqref{METHadditionalflow4}.
\end{proof}

With Proposition \ref{add flow} and Corollary \ref{additionflowsonLnMmAnk}, the following theorem can be proved.

\begin{theorem}\label{symmetry}
The additional flows $\partial_{t^*_{m,l}}$ commute
with the bigraded  Toda hierarchy flows $\partial_{t_{\gamma,n}}$, i.e.,
\begin{eqnarray}
[\partial_{t^*_{m,l}}, \partial_{t_{\gamma,n}}]\Phi=0,
\end{eqnarray}
where $\Phi$ can be $S$, $\bar S$ or $\L$, $\gamma=0,1,2; n\geq 0$ and
 $
\partial_{t^*_{m,l}}=\frac{\partial}{\partial{t^*_{m,l}}},
\partial_{t_{\gamma,n}}=\frac{\partial}{\partial{t_{\gamma,n}}}$.

\end{theorem}
\begin{proof} According to the definition,
\begin{eqnarray*}
[\partial_{t^*_{m,l}},\partial_{t_{\gamma,n}}]S=\partial_{t^*_{m,l}}
(\partial_{t_{\gamma,n}}S)-
\partial_{t_{\gamma,n}} (\partial_{t^*_{m,l}}S),
\end{eqnarray*}
and using the actions of the additional flows and the bigraded Toda
flows on $S$, for  $\gamma=0,1;$, we have
\begin{eqnarray*}
[\partial_{t^*_{m,l}},\partial_{t_{\gamma,n}}]S
&=& -\partial_{t^*_{m,l}}\left((B_{\gamma,n})_{-}S\right)+
\partial_{t_{\gamma,n}} \left(((M-\bar M)^m\L^l)^m_{-}S \right)\\
&=& -(\partial_{t^*_{m,l}}B_{\gamma,n} )_{-}S-
(B_{\gamma,n})_{-}(\partial_{t^*_{m,l}}S)\\&&+
[\partial_{t_{\gamma,n}} ((M-\bar M)^m\L^l)]_{-}S +
((M-\bar M)^m\L^l)_{-}(\partial_{t_{\gamma,n}}S).
\end{eqnarray*}
Using \eqref{definitionadditionalflowsonphi2} and Proposition \ref{flowsofM}, it
equals
\begin{eqnarray*}
[\partial_{t^*_{m,l}},\partial_{t_{\gamma,n}}]S
&=&[\left((M-\bar M)^m\L^l\right)_{-}, B_{\gamma,n}]_{-}S+
(B_{\gamma,n})_{-}\left((M-\bar M)^m\L^l\right)_{-}S\\
&&+[(B_{\gamma,n})_{+},(M-\bar M)^m\L^l]_{-}S-((M-\bar M)^m\L^l)_{-}(B_{\gamma,n})_{-}S\\
&=&[((M-\bar M)^m\L^l)_{-}, B_{\gamma,n}]_{-}S- [(M-\bar M)^m\L^l,
(B_{\gamma,n})_{+}]_{-}S\\&&+
[(B_{\gamma,n})_{-},((M-\bar M)^m\L^l)_{-}]S\\
&=&0.
\end{eqnarray*}
Similarly, using \eqref{definitionadditionalflowsonphi2} and Proposition \ref{flowsofM}, we can prove the additional flows   commute with
flows $\d_{t_{\beta,n}}$ in the sense
of acting on $S$. Of course, using
\eqref{definitionadditionalflowsonphi2}, \eqref{bn1} and Theorem \ref{flowsofM}, we can also prove the additional flows  commute
with all flows of METH  in the sense of acting on $\bar S, L$. Here
we also give the proof for commutativity of additional symmetries with
$\d_{t_{2,n}}$. To be a little
different from the proof above, we let the Lie bracket act on $\bar S$,
\begin{eqnarray*}
[\partial_{t^*_{m,l}},\partial_{t_{2,n}}]\bar S &=&
\partial_{t^*_{m,l}}\left((B_{2,n})_{+}\bar S \right)-
\partial_{t_{2,n}} \left(((M-\bar M)^m\L^l)_{+}\bar S  \right)\\
&=& (\partial_{t^*_{m,l}}B_{2,n} )_{+}\bar S +
(B_{2,n})_{+}(\partial_{t^*_{m,l}}\bar S
)\\&&-(\partial_{t_{2,n}} ((M-\bar M)^m\L^l))_{+}\bar S
-((M-\bar M)^m\L^l)_{+}(\partial_{t_{2,n}}\bar S ),
\end{eqnarray*}
which further leads to
\begin{eqnarray*}
[\partial_{t^*_{m,l}},\partial_{t_{2,n}}]\bar S
&=&[((M-\bar M)^m\L^l)_{+}, B_{2,n}]_{+}\bar S+
(B_{2,n})_{+}((M-\bar M)^m\L^l)_{+}\bar S \\
&&+[(B_{2,n})_{-},(M-\bar M)^m\L^l]_{+}\bar S -((M-\bar M)^m\L^l)_{+}(B_{2,n})_{+}\bar S \\
&=&[((M-\bar M)^m\L^l)_{+}, B_{2,n}]_{+}\bar S
+[(B_{2,n})_{-},(M-\bar M)^m\L^l]_{+}\bar S\\&& +
[(B_{2,n})_{+},((M-\bar M)^m\L^l)_{+}]\bar S \\
&=&[((M-\bar M)^m\L^l)_{+}, B_{2,n}]_{+}\bar S  +
[B_{2,n},((M-\bar M)^m\L^l)_{+}]_{+}\bar S =0.
\end{eqnarray*}
In the proof above, $[(B_{\gamma,n})_{+},
((M-\bar M)^m\L^l)]_{-}= [(B_{\gamma,n})_{+}, ((M-\bar M)^m\L^l)_{-}]_{-}$ has
been used, since $[(B_{\gamma,n})_{+}, ((M-\bar M)^m\L^l)_{+}]_{-}=0$. The other cases in the theorem can be proved in similar ways.
\end{proof}
The commutative property in Theorem \ref{symmetry} means that
additional flows are symmetries of the METH.
Since they are symmetries, it is natural to consider the algebraic
structures among these additional symmetries. So we obtain the following important
theorem.
\begin{theorem}\label{WinfiniteCalgebra}
The additional flows  $\partial_{t^*_{m,l}}(m>0,l\geq 0)$ form a Block type Lie algebra with the
following relation
 \begin{eqnarray}\label{algebra relation}
[\partial_{t^*_{m,l}},\partial_{t^*_{n,k}}]= (km-n l)\d^*_{m+n-1,k+l-1},
\end{eqnarray}
which holds in the sense of acting on  $S$, $\bar S$ or $\L$ and  $m,n\geq0,\, l,k\geq 0.$
\end{theorem}
\begin{proof}
 By using
 (\ref{definitionadditionalflowsonphi2}), we get
\begin{eqnarray*}
[\partial_{t^*_{m,l}},\partial_{t^*_{n,k}}]S&=&
\partial_{t^*_{m,l}}(\partial_{t^*_{n,k}}S)-
\partial_{t^*_{n,k}}(\partial_{t^*_{m,l}}S)\\
&=&-\partial_{t^*_{m,l}}\left(((M-\bar M)^n\L^k)_{-}S\right)
+\partial_{t^*_{n,k}}\left(((M-\bar M)^m\L^l)_{-}S\right)\\
&=&-(\partial_{t^*_{m,l}}
(M-\bar M)^n\L^k)_{-}S-((M-\bar M)^n\L^k)_{-}(\partial_{t^*_{m,l}} S)\\
&&+ (\partial_{t^*_{n,k}} (M-\bar M)^m\L^l)_{-}S+
((M-\bar M)^m\L^l)_{-}(\partial_{t^*_{n,k}} S).
\end{eqnarray*}
On the account of  Proposition\ref{add flow}, we further get
 \begin{eqnarray*}&&
[\partial_{t^*_{m,l}},\partial_{t^*_{n,k}}]S\\
&=&-\Big[\sum_{p=0}^{n-1}
(M-\bar M)^p(\partial_{t^*_{m,l}}(M-\bar M))(M-\bar M)^{n-p-1}\L^k
+(M-\bar M)^n(\partial_{t^*_{m,l}}\L^k)\Big]_{-}S\\&&-((M-\bar M)^n\L^k)_{-}(\partial_{t^*_{m,l}} S)\\
&&+\Big[\sum_{p=0}^{m-1}
(M-\bar M)^p(\partial_{t^*_{n,k}}(M-\bar M))(M-\bar M)^{m-p-1}\L^l
+(M-\bar M)^m(\partial_{t^*_{n,k}}\L^l)\Big]_{-}S\\&&+
((M-\bar M)^m\L^l)_{-}(\partial_{t^*_{n,k}} S)\\
&=&[(nl-km)(M-\bar M)^{m+n-1}\L^{k+l-1}]_-S\\
&=&(km-nl)\d^*_{m+n-1,k+l-1}S.
\end{eqnarray*}
Similarly  the same results on $\bar S$ and $\L$ are as follows
 \begin{eqnarray*}
[\partial_{t^*_{m,l}},\partial_{t^*_{n,k}}]\bar S
&=&((km-nl)(M-\bar M)^{m+n-1}\L^{k+l-1})_+\bar S\\
&=&(km-nl)\d^*_{m+n-1,k+l-1}\bar S,
\\[6pt]
{}[\partial_{t^*_{m,l}},\partial_{t^*_{n,k}}]\L&=&
\partial_{t^*_{m,l}}(\partial_{t^*_{n,k}}\L)-
\partial_{t^*_{n,k}}(\partial_{t^*_{m,l}}\L)\\
&=&[((nl-km)(M-\bar M)^{m+n-1}\L^{k+l-1})_-, \L]\\
&=&(km-nl)\d^*_{m+n-1,k+l-1}\L.
\end{eqnarray*}
\end{proof}
Denote $d_{m,l}=\d_{t^*_{m+1,l+1}}$, and let Block algebra be the span of all $d_{m,l},\,m,l\ge-1$.
Then by \eqref{algebra relation}, Block algebra is a Lie algebra with relations
\begin{eqnarray}
[d_{m,l},d_{n,k}]=((m+1)(k+1)-(l+1)(n+1))d_{m+n,l+k}\mbox{ \ for \ } m,n,l,k\geq -1.
\end{eqnarray}
Thus Block algebra is in fact a Block type Lie algebra which is generated by the set
\begin{eqnarray}\label{Gen-B}
B=\{ d_{-1,0},
d_{0,-1}, d_{0,0}, d_{1,0}, d_{0,1}\}=\{ \d^*_{0,1}, \d^*_{1,0}, \d^*_{1,1}, \d^*_{2,1},
\d^*_{1,2}\}.
\end{eqnarray}

To see the generalized Block flow equation clearly, we give two simple examples, i.e. the $t^*_{1,0},t^*_{1,1}$ flow equation of the METH as following.
The $t^*_{1,0}$ flow equation of the METH
\begin{eqnarray}\label{METHadditionalflow1111}
\dfrac{\partial \L}{\partial
{t^*_{1,0}}}&=&[(M-\bar M)_{+}, \L]=\sum_{n\geq 1}
nt_{0,n}\d_{t_{0,n-1}}+\sum_{n\geq 1}
t_{1,n}\d_{t_{1,n-1}}+\sum_{n\geq 2}t_{2,n}\d_{t_{2,n-1}}-1.
\end{eqnarray}
The $t^*_{1,0}$ flow equation of the METH is as
\begin{eqnarray}\label{METHadditionalflow1111}
\dfrac{\partial \L}{\partial
{t^*_{1,1}}}&=&[(M\L-\bar M\L)_{+}, \L]\\
&=&[\frac{x}{\epsilon},\L]-\L+\sum_{n\geq 1}
nt_{0,n}\d_{t_{0,n}}+\sum_{n\geq 1}
t_{1,n}\d_{t_{1,n}}+\sum_{n\geq 1}t_{2,n}\d_{t_{2,n}}+2(t_{1,n}-t_{2,n})\d_{t_{0,n}},
\end{eqnarray}
which further leads to

\begin{eqnarray}\label{METHadditionalflow1111}
&&\dfrac{\partial u}{\partial
{t^*_{1,1}}}=u+(\sum_{n\geq 1}
nt_{0,n}\d_{t_{0,n}}+\sum_{n\geq 1}
t_{1,n}\d_{t_{1,n}}+\sum_{n\geq 1}t_{2,n}\d_{t_{2,n}}+2(t_{1,n}-t_{2,n})\d_{t_{0,n}})u,\\
&& \dfrac{\partial v}{\partial
{t^*_{1,1}}}=2v+(\sum_{n\geq 1}
nt_{0,n}\d_{t_{0,n}}+\sum_{n\geq 1}
t_{1,n}\d_{t_{1,n}}+\sum_{n\geq 1}t_{2,n}\d_{t_{2,n}}+2(t_{1,n}-t_{2,n})\d_{t_{0,n}})v.
\end{eqnarray}

\begin{theorem}
The Block flows of the modified extended Toda hierarchy  are Hamiltonian systems
of the form
\begin{equation}
\frac{\d u}{\d t^*_{m,l}}=\{u,H^*_{m,l}\}_1,\ \frac{\d v}{\d t^*_{m,l}}=\{v,H^*_{m,l}\}_1,\ \ m,l\geq 0.
\end{equation}
They satisfy the following bi-Hamiltonian recursion relation
\[ \frac{\d }{\d t^*_{m,l}}=\{\cdot,H^*_{m,l-1}\}_2=n
\{\cdot,H^*_{m,l}\}_1.
\]
Here the Hamiltonians (depending on $t_{\alpha,n}$) with respect to $t^*_{m,l}$ have the form
\begin{equation}
H^*_{m,l}=\int h^*_{m,l}(u,v; u_x,v_x; \dots;t_{\beta,n}; \epsilon) dx,\quad \beta=0,1,2; \ n\ge 0,
\end{equation}
with the Hamiltonian densities  $ h^*_{m,l}(u,v; u_x,v_x; \dots;t_{\alpha,n}; \epsilon)$ given by
\[
 h^*_{m,l}&=&\res (M-\bar M)^m\L^l,
\]

\end{theorem}
\begin{proof}
The proof is similar as the proof for original Toda flows.
\end{proof}

\section{Conclusion and discussion}

In this paper, because of the universality of additional symmetry for integrable systems and obstacles in the construction of the additional symmetry  induced by the extended flows in the ETH with a single logarithm which is a combination of $\log_+ \L$ and $\log_- \L$,  we modified the ETH and construct one new kind of integrable system named as the METH with two separated logarithmic flows. This new hierarchy has bi-Hamiltonian structure and tau symmetry. Under the modification, this new integrable hierarchy was proved to have a nice Block type additional symmetry. Because of the important application of ETH in topological field theory, the application of this new hierarchy on topological fields theory should be an interesting subject.

{\bf {Acknowledgements:}}
  Chuanzhong Li is supported by the National Natural Science Foundation of China under Grant No. 11201251, the Natural Science Foundation of Zhejiang Province under Grant No. LY12A01007,
 the Natural Science Foundation of Ningbo under Grant No. 2013A610105. Jingsong He is supported by the National Natural Science Foundation of China under Grant No. 11271210, K.C.Wong Magna Fund in
Ningbo University.

\end{document}